\documentclass{article}
\usepackage{amsmath}
\usepackage{fancyhdr}
\usepackage{amssymb}
\usepackage{amsthm}
\usepackage{graphicx}
\usepackage{varioref}
\usepackage{verbatim} 
\usepackage{multicol}
\usepackage{enumerate}
\usepackage[normalem]{ulem}
\usepackage{caption}
\usepackage{subcaption}
\usepackage[T1]{fontenc}
\usepackage[margin=1in]{geometry}
\usepackage{thm-restate}
\usepackage{enumitem}

\usepackage{mathrsfs}

\usepackage{url}
\usepackage{xspace}
\usepackage{thm-restate}

\usepackage{mathpazo}

\usepackage{microtype}

\usepackage{tikz}
\usetikzlibrary{calc}
\usetikzlibrary{decorations.pathmorphing}
\usetikzlibrary{decorations.markings}

\usepackage{epstopdf}

\usepackage[colorlinks = true]{hyperref}
\usepackage{xcolor}
\definecolor{darkred}  {rgb}{0.5,0,0}
\definecolor{darkblue} {rgb}{0,0,0.5}
\definecolor{darkgreen}{rgb}{0,0.5,0}
\hypersetup{
  urlcolor   = blue,         
  linkcolor  = darkblue,     
  citecolor  = darkgreen,    
  filecolor  = darkred       
}

\usepackage{cleveref}
\crefname{lemma}{Lemma}{Lemmas}
\crefname{proposition}{Proposition}{Propositions}
\crefname{definition}{Definition}{Definitions}
\crefname{theorem}{Theorem}{Theorems}
\crefname{conjecture}{Conjecture}{Conjectures}
\crefname{corollary}{Corollary}{Corollaries}
\crefname{section}{Section}{Sections}
\crefname{appendix}{Appendix}{Appendices}
\crefname{figure}{Fig.}{Figs.}
\crefname{equation}{Eq.}{Eqs.}
\crefname{table}{Table}{Tables}
\crefname{claim}{Claim}{Claims}

\newlist{properties}{enumerate}{10}
\setlist[properties]{label*=\arabic*}

\crefname{propertiesi}{Effective Resistance Property}{Effective Resistance Properties}

\newlist{propertees}{enumerate}{10}
\setlist[propertees]{label*=\roman*}
\crefname{properteesi}{Effective Capacitance Property}{Effective Capacitance Properties}

\newtheorem{theorem}{Theorem}
\newtheorem{lemma}[theorem]{Lemma}

\newtheorem{definition}[theorem]{Definition}

\newtheorem*{conjecture*}{Conjecture}

\theoremstyle{definition}

\usepackage{authblk}

\newcommand{\ket}[1]{|#1\rangle}

\newcommand{\tth}[0]{\textsuperscript{th}}

\DeclareMathAlphabet{\matheu}{U}{eus}{m}{n}

\newcommand{\sop}[1]{{\mathcal #1}}

\newcommand{\ketbra}[2]{|{#1}\rangle\!\langle{#2}|}

\newcommand{\Gcyc}{{G_\textrm{cyc}}}
\newcommand{\Gd}[1]{{K^G_{#1}}}
\newcommand{\Gbi}{{G_\textrm{bip}}}

\begin{document}

\title{Applications of the quantum algorithm for st-connectivity}
\author[1]{Kai DeLorenzo}
\author[1]{Shelby Kimmel}
\author[1]{R. Teal Witter}
\affil[1]{Middlebury College, Computer Science Department, Middlebury, VT}

\date{}

\maketitle

\begin{abstract}
We present quantum algorithms for various problems related to graph connectivity. We give simple and query-optimal algorithms for cycle detection and odd-length cycle detection (bipartiteness) using a reduction to st-connectivity.  Furthermore, we show that our algorithm for cycle detection has improved performance under the promise of large circuit rank or a small number of edges. We also provide algorithms for detecting even-length cycles and for estimating the circuit rank of a graph. All of our algorithms have logarithmic space complexity. \end{abstract}

\section{Introduction}

Quantum query algorithms are remarkably described by span programs
\cite{reichardt2009span,R01}, a linear algebraic object originally created to
study classical logspace complexity \cite{KW93}. However, finding optimal span
program algorithms can be challenging; while they can be obtained using a
semidefinite program, the size of the program grows exponentially with the
size of the input to the algorithm. Moreover, span programs are designed to
characterize the query complexity of an algorithm, while in practice we also
care about the time and space complexity.

One of the nicest span programs is for the problem of undirected
$st$-connectivity, in which one must decide whether two vertices $s$ and $t$
are connected in a given graph. It is ``nice'' for several reasons:
\begin{itemize}
\item It is easy to describe and understand why it is correct.
\item It corresponds to a quantum algorithm that uses logarithmic (in the number of vertices and edges of the graph) space \cite{BR12,Jeffery2017algorithmsgraph}.
\item The time complexity of the corresponding algorithm is the product of the
query complexity, and the time required to implement a unitary that applies
one step of a quantum walk on the underlying graph \cite{BR12,Jeffery2017algorithmsgraph}.
On the complete graph, for example, the quantum walk step introduces only an additional
logarithmic factor to the complexity \cite{BR12}.
\item The query complexity of the algorithm is determined by two well known graph functions, the effective resistance and effective capacitance \cite{JJKP2018}.
\end{itemize}

Thus one strategy for designing other ``nice'' quantum algorithms is to reduce
a given problem to $st$-connectivity, and then use the span program
$st$-connectivity algorithm. This strategy has proven to be quite successful,
and in fact has produced several optimal or nearly optimal algorithms. There
is a reduction from Boolean formula evaluation to $st$-connectivity
\cite{Nisan:1995:SLC:225058.225101} that produces an optimal quantum algorithm
for read-once formulas \cite{Jeffery2017algorithmsgraph}. There is an optimal
reduction from graph connectivity to $st$-connectivity \cite{JJKP2018}. Cade
et al. use an $st$-connectivity subroutine to create nearly query-optimal
algorithms for cycle detection and bipartiteness
\cite{cade2016time}\footnote{In Ref. \cite{Arins2016}, {\={A}}ri{\c{n}}{\v{s}}
designed algorithms for connectivity and bipartiteness which, while not
strictly reductions to $st$-connectivity, are very closely related.}. Finally,
the $st$-connectivity span program algorithm underlies the learning graph framework \cite{Bel11}, one of
the most successful heuristics for span program algorithm design.

In this work, we follow precisely this strategy for creating ``nice'' quantum
algorithms: we reduce the graph problems of cycle detection, odd-length path
detection, bipartiteness, and even-length cycle detection to
$st$-connectivity. In our reductions, solving the related $st$-connectivity
problem comprises the whole algorithm; in other words, we create a new graph
that has an $st$-path if and only if the original graph has the property in
question.

Additionally, there is an an estimation algorithm closely related to the
$st$-connectivity algorithm that determines the size of the effective
resistance or effective capacitance of the graph
\cite{ito2015approximate,JJKP2018}. Not only is it often useful to estimate
the effective resistance or effective capacitance of a graph, as these
quantities bound the shortest path length and smallest cut size respectively,
but sometimes one can encode quantities of interest as either the effective
capacitance or effective resistance. For example, given a graph $G$, it is
possible to create a new graph whose effective resistance is the average
effective resistance (Kirchoff index) of the original graph \cite{JJKP2018}.
Using this strategy of reduction to effective resistance estimation, we also create an algorithm to estimate the circuit rank of a graph.

\subsection{Contributions and Comparison to Previous Work}

All of our algorithms are in the adjacency matrix model (see \cref{section:background}), in which one can query elements of the adjacency matrix of the input graph. This contrasts with work such as \cite{cade2016time} which study similar problems in the adjacency list model. 

We note all of our algorithms are space efficient, in that the number of qubits
required are logarithmic in the number of edges and vertices of the graph.
(This property is inherited directly from the basic $st$-connectivity span
program.) We do not analyze time complexity, but, as mentioned above, it is
the product of the query complexity, which we analyze, and the time
required to implement certain quantum walk unitaries. (We leave this analysis
for future work.) 

We next discuss the context of each of our results in turn.
In this section, we assume the underlying graph is the complete graph on $n$ vertices to
more easily compare to previous work, although in the main body of the paper,
we show our results  apply more to generic underlying graphs with $n$
vertices and $m$ edges.

\paragraph{Cycle Detection} Cade et al. \cite{cade2016time} describe a nearly optimal $\tilde{O}(n^{3/2})$ quantum
query algorithm for cycle detection via reduction to $st$-connectivity, almost matching the lower bound of
$\Omega(n^{3/2})$ \cite{childs2012quantum}. In this work, we find an
algorithm that removes the log-factors of the previous result, giving an algorithm with optimal
$O(n^{3/2})$ query complexity. Moreover, our algorithm is simpler than that in Ref.
\cite{cade2016time}: their approach requires solving an $st$-connectivity problem
 within a Grover search, while our approach is entirely based on solving
an $st$-connectivity problem.

We furthermore prove that if promised that (in the case of a cycle) the circuit
rank of the graph is at least $r$ or (in the case of no cycles) there are at
most $\mu$ edges, then the query complexity of cycle detection is
$O(\mu\sqrt{n/r})$.

\paragraph{Bipartiteness} An optimal quantum query algorithm for bipartiteness
was created by {\={A}}ri{\c{n}}{\v{s}} \cite{Arins2016}, matching the lower bound of
$\Omega(n^{3/2})$ \cite{zhang2005power}. However, this algorithm was not known
to be time efficient (and did not use a reduction to $st$-connectivity,
although the ideas are quite similar to the approach here and in
\cite{cade2016time}). In Ref. \cite{cade2016time}, Cade et al., using similar ideas as in their
cycle detection algorithm, create a bipartiteness checking algorithm using a
reduction to $st$-connectivity that is again embedded in a search loop, which
is optimal up to logarithmic factors in query complexity. Our algorithm for bipartiteness removes
the logarithmic factors of \cite{cade2016time}, and so recovers the optimal query complexity of \cite{Arins2016},
while retaining the simplicity of the reduction to $st$-connectivity. 

\paragraph{Even-length Cycle Detection} The problem of detecting an even-length cycle
can provide insight into the structure of the graph. For example, it is
straightforward to see that in a graph with no even cycles, no edge can be
involved in more than one cycle. Classically this problem requires
$\Theta(n^2)$ queries \cite{yuster1997finding}. We are not aware of an
existing quantum algorithm for this problem; we provide an $O(n^{3/2})$ query
algorithm.

\paragraph{Estimation of Circuit Rank} Circuit rank parameterizes the number
of cycles in a graph: it is the number of edges that must be removed before
there are no cycles left in a graph. It has also been used to describe the
complexity of computer programs \cite{mccabe1976complexity}. We give an
algorithm to estimate the circuit rank $r$ to multiplicative error $\epsilon$
with query complexity $\widetilde{O}\left(\epsilon^{-3/2}\sqrt{n^4/r}\right)$
in the generic case and query complexity
$\widetilde{O}\left(\epsilon^{-3/2}\sqrt{n^3/r}\right)$ when promised that the
graph is a cactus graph. When additionally promised that the circuit rank is
large, these algorithms can have non-trivial query complexity. We are aware of no
other classical or quantum query algorithms that determine or estimate this
quantity.

\paragraph{Odd-length Path} We provide an algorithm to determine whether
there is an odd-length path between two specified vertices that uses
$O(n^{3/2})$ queries. While perhaps not the most interesting problem on its own, we effectively
leverage this algorithm as a subroutine in several of our other constructions.

\subsection{Open Problems}

Throughout this paper, our strategy is to take a graph $G$, use it to create a
new graph $G'$, such that there is an $st$-path through $G'$ if and only if
$G$ has a certain property. We analyze the query complexity of the algorithms we create in detail, but not the time complexity. The
time complexities of our algorithms depend on the time required to implement
one step of a quantum walk on the graph $G'$ (see $\mathsf{U}$ from \cref{thm:stconn} and Ref. \cite{Jeffery2017algorithmsgraph}) \footnote{In Ref. \cite{cade2016time}, they claim that their algorithms are time efficient, but their time analysis only considers the case of the underlying graph being the complete graph, while the actual graph used in their algorithms is not the complete graph. We expect that their algorithms can be implemented efficiently, but more work is needed to show this.}. We strongly suspect that the highly structured nature of the graphs $G'$ we consider would yield time-efficient algorithms, but we have not done a full analysis. 

In the case of our algorithm for cycle detection, we create a graph $G'$ whose
effective resistance is the circuit rank of the original graph $G$. For the
graphs we design for the other algorithms in this paper, do the effective
resistance or effective capacitance have relevant meanings?

The query complexity of our algorithm for estimating the circuit rank of a graph
depends on bounding a quantity called the approximate negative witness. While
the best bound we currently have is $O(n^4)$, we believe this is not tight.
Obtaining a better bound would not only be interesting for these results, but
could provide insight into more general quantum estimation algorithms.

Several of our results rely on a graph that tests for paths of odd length,
i.e. those whose length modulo 2 is 1. Is there a way to adapt our algorithm
to test for paths of arbitrary modulus?

Ref. \cite{alvarez2000compendium} provides a list of complete problems for symmetric logarithmic space (SL), of which $st$-connectivity is one such problem. It would be interesting to study the query complexity of these problems to see if reducing to $st$-connectivity always gives an optimal approach.

Finally, it would be nice to improve the quantum lower bounds and classical bounds for several of these
problems. In particular, we would like to obtain better quantum lower bounds for the even-length cycle detection and odd-length path
detection problems. (For the later, the best quantum lower bound is $\Omega(n)$
\cite{BR12}.) We expect that the promise of large circuit rank also aids a classical algorithm for cycle detection, and it would be interesting to know by how much, in order to compare to our quantum algorithm.

\section{Preliminaries}\label{section:background}

We consider undirected graphs $G=(V,E)$ where $V$ is a set of vertices and $E$
a set of edges; we often use $E(G)$ and $V(G)$ to denote the sets of edges and
vertices of a graph $G$ when there are multiple graphs involved. If
clear which graph we are referring to, we will use $n$ for the number of
vertices in the graph, and $m$ for the number of edges. For ease of notation,
we associate each edge with a unique label $\ell$. For example, we refer to an
edge between vertices $u$ and $v$ labeled by $\ell$ as $\{u,v\}$, or simply as
$\ell$. 
In general, we could consider a weighting function on the edges or consider graphs with multi-edges (see \cite{JJKP2018} for more details on how these modifications are implemented) but for our purposes, we will always consider all edges to have weight $1$, and we will only consider graphs with a at most a single edge between any two vertices.

We will use the following notation regarding spanning trees: if $G$ is a
connected graph and $\ell\in E(G)$, then we use $t_\ell(G)$ to denote the
number of unique spanning trees of $G$ that include edge $\ell$, and we use
$t(G)$ to denote the total number of unique spanning trees of $G.$

In \cref{sec:Cycle} we describe a quantum algorithm for estimating the circuit rank of a graph, which is a quantity that is relevant for a number of applications, like determining the robustness of a network, analyzing chemical structure \cite{doi:10.1021/ci60005a013}, or parameterizing the complexity of a program \cite{mccabe1976complexity}.
\begin{definition}[Circuit Rank]\label{def:circRank}
The \emph{circuit rank} of a graph with $m$ edges and $n$ vertices is $m-n+\kappa$ where $\kappa$ is the
number of connected components. Alternatively, the circuit rank of a graph is
the minimum number of edges that must be removed to break all cycles and turn
the graph into a tree or forest.
\end{definition}

We also consider a special class of graphs called cactus graphs:
\begin{definition}
A cactus graph is a connected graph in which any two simple cycles share at most one common vertex.
\end{definition}
\noindent We will in particular use cactus forests, in which all components of
the graph are cacti. For cactus forests, the circuit rank is simply the total
number of cycles in the graph.

The final type of graph we need is the bipartite double graph:
\begin{definition}
Given a graph $G$, the bipartite double of $G$, denoted $K^G$, is the graph that consists of two copies of the vertices of $G$ (with vertex $v\in V(G)$ labeled as $v_0$ in the first copy and $v_1$ in the second), with all original edges removed, and edges $\{u_0,v_1\}$ and $\{v_0,u_1\}$ created for each edge $\{u,v\}\in E(G).$ This graph is also known as the Kronecker cover of $G$, or $G\times K_2$, and its adjacency matrix is given by $G\otimes X$ (where $X$ is the Pauli $X$ operator.)
\end{definition}

We associate edges of $G$ with literals of a string $x\in \{0,1\}^N$, where
$N\leq |E|$, and a literal is either $x_i$ or $\overline{x}_i$ for $i\in[N]$.  
The subgraph $G(x)$ of $G$ contains an edge $\ell$ if $\ell$
is associated with $x_i$ and $x_i=1$, or if $\ell$ is
associated with $\overline{x}_i$ and $x_i=0$. (See
\cite{JJKP2018} for more details on this association.) 

We assume that we have complete knowledge of $G$, and access to
$x\in\{0,1\}^N$ via a black box unitary (oracle) $O_x$. This oracle acts as
$O_x\ket{i}\ket{b}=\ket{i}\ket{b\oplus x_i}$, where $x_i$ is the $i\tth$ bit
of $x$. Then our goal is to use $O_x$ as few times as possible to determine a property of the graph
$G(x).$ The number of uses of $O_x$ required for a given application is called the query complexity.

We will be applying and analyzing an algorithm for $st$-connectivity, which is
the problem of deciding whether two nodes $s,t\in V(G)$ are connected in a
graph $G(x)$, where $G$ is initially known, but $x$ must be determined using
the oracle, and we are promised $x\in X$, where $X\subseteq\{0,1\}^N$. An
$st$-path is a series of edges connecting $s$ to $t$.

Given two instances of $st$-connectivity, they can be combined in
\emph{parallel}, where the two $s$ vertices are identified and labeled as the
new $s$, and the two $t$ vertices are identified and labeled as the new $t$.
This new $st$-connectivity problem encodes the logical \textsc{or} of the
original connectivity problems, in that the new graph is connected if and only
if at least one of the original graphs was connected. Two instances of
$st$-connectivity can also be combined in \emph{series}, where the $s$ vertex
of one graph is identified with the $t$ vertex of the other graph, and
relabeled using a label not previously used for a vertex in either graph.
This new $st$-connectivity encodes the logical \textsc{and} of the original
connectivity problems, in that the new graph is connected if and only if both
of the original graphs were connected. This correspondence was noted in
\cite{Nisan:1995:SLC:225058.225101} and applied to quantum query algorithms
for Boolean formulas in \cite{Jeffery2017algorithmsgraph} and for determining
total connectivity in \cite{JJKP2018}.

The key figures of merit for determining the query complexity of the
span-program-based $st$-connectivity quantum algorithm are effective
resistance and effective capacitance \cite{JJKP2018}. We use $R_{s,t}(G(x))$
to denote the effective resistance between vertices $s$ and $t$ in a graph
$G(x)$, and we use $C_{s,t}(G(x))$ to denote the effective capacitance between
vertices $s$ and $t$ in $G(x)$. These two functions were originally formulated to
characterize electrical circuits, but are also important functions in graph
theory. For example, effective resistance is related to the hitting time of a
random walk on a graph. (See \cite{bollobas2013modern} for more information on
effective resistance and capacitance in the context of graph theory.) For this
paper, we don't require a formal definition of these functions, but can
instead use a few of their well known and easily derived properties (see
\cite{JJKP2018} for formal definitions):

\paragraph{Effective Resistance}
\begin{properties}
\item If $G$ consists of a single edge between $s$ and $t$, and $G(x)=G$, then $R_{s,t}(G(x))=1$.\label{prop:oneEdge}
\item If $s$ and $t$ are not connected in $G(x)$, then $R_{s,t}(G(x))=\infty$.\label{prop:ERinf}
\item If $G$ consists of subgraphs $G_1$ and $G_2$ connected in series as described above, then\\ $R_{s,t}(G(x))=R_{s,t}(G_1(x))+R_{s,t}(G_2(x))$.\label{prop:ERseries}
\item If $G$ consists of subgraphs $G_1$ and $G_2$ connected in parallel as described above, then\\ $\left(R_{s,t}(G(x))\right)^{-1}=\left(R_{s,t}(G_1(x))\right)^{-1}+\left(R_{s,t}(G_2(x))\right)^{-1}$.\label{prop:ERparallel}
\item If $G(x)$ is a subgraph of $G(y)$, then $R_{s,t}(G(x))\geq R_{s,t}(G(y))$. \label{prop:ERsubgraph}
\item If $s$ and $t$ are connected in $G(x)$, then $R_{s,t}(G(x))\leq d$, where $d$
is the length of the shortest path from $s$ to $t.$\label{prop:length}
\end{properties}

\paragraph{Effective Capacitance}
\begin{propertees}
\item If $G$ consists of a single edge between $s$ and $t$, and $G(x)$ does not include the edge $\{s,t\}$, then $C_{s,t}(G(x))=1$.\label{prop:noEdge}
\item If $s$ and $t$ are connected in $G(x)$, then $C_{s,t}(G(x)=\infty$.\label{prop:ECinf}
\item If $G$ consists of subgraphs $G_1$ and $G_2$ connected in series as described above, then\\ $\left(R_{s,t}(G(x))\right)^{-1}=\left(R_{s,t}(G_1(x))\right)^{-1}+\left(R_{s,t}(G_2(x))\right)^{-1}$.\label{prop:ECseries}
\item If $G$ consists of subgraphs $G_1$ and $G_2$ connected in parallel, then\\ $R_{s,t}(G(x))=R_{s,t}(G_1(x))+R_{s,t}(G_2(x))$.\label{prop:ECparallel}
\item If $G(x)$ is a subgraph of $G(y)$, then $C_{s,t}(G(x))\leq C_{s,t}(G(y))$.\label{prop:ECsubgraph}
\item If $s$ and $t$ are not connected in $G(x)$, $C_{s,t}(G(x))$ is less than the size of the smallest cut in $G$ between $s$ and $t$.\label{prop:ECcut}
\end{propertees}

We analyze our algorithms using these properties rather than first principles to demonstrate the
relative ease of bounding the query complexity of span program algorithms for $st$-connectivity
problems.

Now we can describe the performance of the span program algorithm for deciding $st$-connectivity:
\begin{theorem}\cite{JJKP2018}\label{thm:stconn}
Let $G=(V,E)$ be a graph with $s,t\in V(G)$. Then there is a span program algorithm whose bounded-error
quantum query complexity of evaluating whether $s$ and $t$ are connected in $G(x)$ promised $x\in X$ and $X\subseteq\{0,1\}^N$ is
\begin{align}\label{equation:ourBound}
O\left(\sqrt{\max_{\substack{x\in X \\ R_{s,t}(G(x))\neq\infty}}R_{s,t}(G(x))\times\max_{\substack{x\in X \\ C_{s,t}(G(x))\neq\infty}}C_{s,t}(G(x))}\right). 
\end{align}
Furthermore, the space complexity is
\begin{align}
O(\max\{\log(|E|),\log(|V|)\}),
\end{align}
and the time complexity is $\mathsf{U}$ times the query complexity,
where $\mathsf{U}$ is the time required to perform one step of a quantum walk on $G$ (see \cite{Jeffery2017algorithmsgraph}).
\end{theorem}

Ito and Jeffery describe an algorithm that can be used to estimate $R_{s,t}(G(x))$. It depends on a quantity called the negative witness size, which we denote $\widetilde{R}_-(x,G)$ (this is the quantity $\widetilde{w}_-
(x)$ of \cite{ito2015approximate} tailored to the case of the $st$-connectivity span program). Let $\sop L(U,\mathbb{R})$ be the set of linear maps from a set $U$ to $\mathbb{R}$. Then we have the following definition
\begin{definition}[See Theorem 4.2, \cite{ito2015approximate}]
Let $G=(V,E)$ with $s,t\in V$. If $G(x)$ is connected from $s$ to $t$, let $V_x\subseteq V$ be the set of vertices connected to both $s$ and $t$. Then there is a unique map $\sop V_x\in \sop L(V_x,\mathbb{R})$ such that $\sop V_x(s)=1$, $\sop V_x(t)=0$, and $\sum_{\{u,v\}\in E(G(x))}\left(\sop V_x(u)-\sop V_x(v)\right)^2$ is minimized. Then the negative approximate witness size of input $x$ on the graph $G$ is
\begin{align}\label{eq:NegAppWitdef}
\widetilde{R}_-(x,G)=\min_{\sop V\in \sop L(V,\mathbb{R}):\sop V(u)=\sop V_x(u) \textrm{ if }u\in V_x}\sum_{\{u,v\}\in E}\left(\sop V(u)-\sop V(v)\right)^2.
\end{align}

\end{definition}

\begin{theorem}\label{thm:witSizeEst}
Let $G$ be a graph with $s,t\in V(G)$. Then the bounded-error quantum query
complexity of estimating $R_{s,t}(G(x))$ to multiplicative error $\epsilon$
promised $s$ and $t$ are connected in $G(x)$ and $x\in X$ for $X\subseteq\{0,1\}^N$ is
$\widetilde{O}\left(\epsilon^{-3/2}\sqrt{R_{s,t}(G(x))\widetilde{R}_-}\right)$,
where $\widetilde{R}_-=\max_{x\in X}\widetilde{R}_-(x,G).$
\end{theorem}


\section{Quantum Algorithms for Detecting and Characterizing Cycles}\label{sec:Cycle}

In this section, we prove the following results on detecting and characterizing cycles:
\begin{theorem}\label{thm:cycleDetection}
Let $G$ be the complete graph on $n$ vertices. If we are promised that either
$G(x)$ is connected with circuit rank at least $r$, or $G(x)$ is not connected
and contains at most $\mu$ edges, then the bounded-error quantum query
complexity of detecting a cycle in $G(x)$ is $O(\mu\sqrt{n/r})$.
\end{theorem}

\begin{theorem}\label{thm:circuitRankEst}
Given a generic graph $G$ with $m$ edges, and a parameter $\epsilon\ll 1$
(here $\epsilon$ can be a constant or can depend on the input) there is a
quantum algorithm that estimates the circuit rank $r$ of $G(x)$ to
multiplicative error $\epsilon$ using
$\widetilde{O}\left(\epsilon^{-3/2}\sqrt{m\mu/r}\right)$ applications of $O_x$, under the promise that $G(x)$ has at most $\mu$ edges.
\end{theorem}

A few notes on these theorems:
\begin{itemize}
\item \cref{thm:cycleDetection} has a worst case upper bound of $O(n^{3/2})$, which matches the optimal lower bound. This is because $r\geq 1$ ($r$ takes value $1$ in the case of a single cycle) and $\mu\leq n-1$ (since a graph without cycles must be a forest). 
\item In \cref{thm:circuitRankEst}, if $r$ and $\epsilon$ are $O(1)$ and if nothing is known about $\mu$ (in which case it could be as large as $m$), one would need to query all edges of the graph. However, given a promise that $r$ is large, for example if $r=\Omega(m^\beta)$ for a positive constant $\beta$, or a promise on $\mu$, we can do better than the trivial classical algorithm of querying all edges.

\end{itemize}

We prove both of these results using a reduction from cycle detection to $st$-connectivity. Specifically we construct a graph $\Gcyc$ such that $\Gcyc(x)$ has an
$st$-path if and only if $G(x)$ has a cycle. We note that there is a cycle in
$G(x)$ if and only if an edge $\{u,v\}$ is present in $G(x)$ and there is a
path from $u$ to $v$ in $G(x)$ that does not use the edge $\{u,v\}.$ Thus, our
reduction tests every edge in $G$ to determine whether these two conditions
are satisfied. We use the encoding of logical \textsc{and} and \textsc{or} into
$st$-connectivity using serial and parallel composition, as described in \cref{section:background} and in Refs.
\cite{Nisan:1995:SLC:225058.225101,Jeffery2017algorithmsgraph}.

We now describe how to build up $\Gcyc$ from simpler graphs. For an edge
$\ell=\{u,v\}\in E(G)$ let $G^-_{\ell}$ be the graph that is the same as $G$,
except with the edge $\ell$ removed, and the vertex $u$ labeled as $s$, and
the vertex $v$ labeled as $t$. (The choice of which endpoint of $\ell$ is $s$ and which is $t$ is arbitrary - either choice is acceptable). Each edge in $G^-_{\ell}$ is
associated with the same literal of $x$ as the corresponding edge in $G$. Thus
there is an $st$-path in $G^-_\ell(x)$ if and only if there is a path between
$u$ and $v$ in $G(x)$ that does not go through $\{u,v\}$.

Next, for an edge $\ell\in E(G)$ let $G^{1}_{\ell}$ be the graph
with exactly two vertices labeled $s$ and $t$, and one edge between them. The
one edge in $G^{1}_{\ell}$ is associated with the same literal bit of $x$ as
$\ell$. Thus there is an $st$-path in $G^{1}_{\ell}(x)$ if and only if $\ell\in E(G(x))$.

Next, we create the graph $G_\ell$ by connecting $G^{1}_{\ell}$ and $G^-_\ell$ in series, while
leaving the associations between edges and literals the same. Then because connecting $st$-connectivity graphs in series is equivalent to logical \textsc{and},
there is an $st$-path through $G_\ell(x)$ if and only if there is a cycle in $G(x)$
passing through $\ell$. 

Finally, we create the graph $\Gcyc$ by connecting all of the graphs $G_\ell$
(for each $\ell\in E(G)$) in parallel, again retaining the association between
edges and literals. Since attaching graphs in parallel is equivalent to
logical \textsc{or}, $\Gcyc(x)$ has an $st$-path if and only if there is a
cycle through some edge of $G(x)$. See \cref{fig:Gcyc(x)} for an example of the construction of $\Gcyc.$

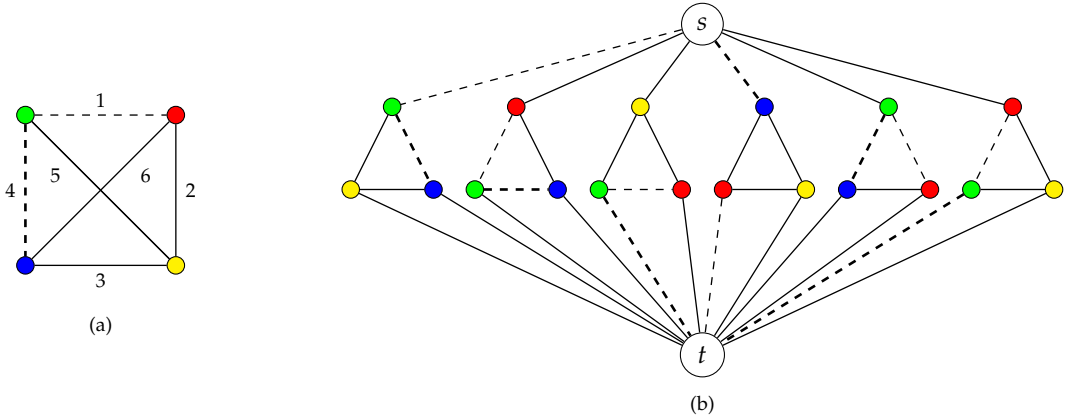
\begin{figure}[ht]
    \centering
    \begin{tikzpicture}[scale = 2]
\node at (0,0) {\begin{tikzpicture}[scale = 2]
            \tikzstyle{vertex}=[draw,circle]
            \tikzstyle{label}=[circle,scale = .75]
            \tikzstyle{edge}=[draw,line width = .5pt,-,black!100]
            \tikzstyle{one}=[draw,line width = .5pt,dashed,black!100]
            \tikzstyle{four}=[draw,line width = 1pt,dashed,black!100]
           	\tikzstyle{vu}=[draw,circle,scale=.7,fill=green]
           	\tikzstyle{vv}=[draw,circle,scale=.7,fill=red]
           	\tikzstyle{vx}=[draw,circle,scale=.7,fill=blue]
           	\tikzstyle{vw}=[draw,circle,scale=.7,fill=yellow]
            \node[vu] (u) at (0,1) {};
            \node[vv] (v) at (1,1) {};
            \node[vw] (w) at (1,0) {};
            \node[vx] (x) at (0,0) {};
            \node[label] (one) at (.5,1.1) {$1$};
            \node[label] (two) at (1.1,.5) {$2$};
            \node[label] (three) at (.5,-.1) {$3$};
            \node[label] (four) at (-.1,.5) {$4$};
            \node[label] (five) at (.2,.6) {$5$};
            \node[label] (six) at (.8,.6) {$6$};
            \draw[one] (u) -- (v);
            \draw[edge] (u) -- (w);
            \draw[edge] (u) -- (w);
            \draw[four] (u) -- (x);
            \draw[edge] (v) -- (x);
            \draw[edge] (v) -- (w);
            \draw[edge] (x) -- (w);
            \node[label] (a) at (.5,-.4) {(a)};
        \end{tikzpicture}};
       \node at (4,0) {
        \begin{tikzpicture}[scale = 2.2]
            \tikzstyle{label}=[circle,scale = .75]
            \tikzstyle{vertex}=[draw,circle,scale=.7]
           	\tikzstyle{vu}=[draw,circle,scale=.7,fill=green]
           	\tikzstyle{vv}=[draw,circle,scale=.7,fill=red]
           	\tikzstyle{vx}=[draw,circle,scale=.7,fill=blue]
           	\tikzstyle{vw}=[draw,circle,scale=.7,fill=yellow]
            \tikzstyle{vertex2}=[draw,circle,scale = 1]
            \tikzstyle{edge}=[draw,line width = .5pt,-,black!100]
            \tikzstyle{one}=[draw,line width = .5pt,dashed,black!100]
            \tikzstyle{four}=[draw,line width = 1pt,dashed,black!100]
            \node[vertex2] (s) at (2.125, 1.5) {$s$};
            \node[vertex2] (t) at (2.125, -.5) {$t$};
            \node[vu] (u1) at (.25,1) {};
            \node[vw] (w1) at (0,.5) {};
            \node[vx] (x1) at (.5,.5) {};
            \draw[edge] (u1) -- (w1);
            \draw[four] (u1) -- (x1);
            \draw[edge] (x1) -- (w1);
            \node[vv] (v2) at (1,1) {};
            \node[vu] (u2) at (.75,.5) {};
            \node[vx] (x2) at (1.25,.5) {};
            \draw[one] (v2) -- (u2);
            \draw[edge] (v2) -- (x2);
            \draw[four] (u2) -- (x2);
            \node[vw] (w3) at (1.75,1) {};
            \node[vu] (u3) at (1.5,.5) {};
            \node[vv] (v3) at (2,.5) {};
            \draw[edge] (w3) -- (u3);
            \draw[edge] (w3) -- (v3);
            \draw[one] (u3) -- (v3);
            \node[vx] (x4) at (2.5,1) {};
            \node[vv] (v4) at (2.25,.5) {};
            \node[vw] (w4) at (2.75,.5) {};
            \draw[edge] (x4) -- (v4);
            \draw[edge] (x4) -- (w4);
            \draw[edge] (v4) -- (w4);
            \node[vu] (u5) at (3.25,1) {};
            \node[vx] (x5) at (3,.5) {};
            \node[vv] (v5) at (3.5,.5) {};
            \draw[four] (u5) -- (x5);
            \draw[one] (u5) -- (v5);
            \draw[edge] (x5) -- (v5);
            \node[vv] (v6) at (4,1) {};
            \node[vu] (u6) at (3.75,.5) {};
            \node[vw] (w6) at (4.25,.5) {};
            \draw[one] (v6) -- (u6);
            \draw[edge] (v6) -- (w6);
            \draw[edge] (w6) -- (u6);
            \draw[one] (s) -- (u1);
            \draw[edge] (s) -- (v2);
            \draw[edge] (s) -- (w3);
            \draw[four] (s) -- (x4);
            \draw[edge] (s) -- (u5);
            \draw[edge] (s) -- (v6);
            \draw[edge] (w1) -- (t);
            \draw[edge] (x1) -- (t);
            \draw[edge] (u2) -- (t);
            \draw[edge] (x2) -- (t);
            \draw[four] (u3) -- (t);
            \draw[edge] (v3) -- (t);
            \draw[one] (v4) -- (t);
            \draw[edge] (w4) -- (t);
            \draw[edge] (x5) -- (t);
            \draw[edge] (v5) -- (t);
            \draw[four] (u6) -- (t);
            \draw[edge] (w6) -- (t);
             \node[label] (b) at (2.125,-.8) {(b)};
        \end{tikzpicture}};
        \end{tikzpicture}
    \caption{(a) A graph $G(x)$, where edges 2, 3, 5, and 6 are present (solid lines indicate the presence of an edge, dashed lines indicate the absence of an edge). (b) The graph $\Gcyc(x)$ that $G(x)$ produces. There is a cycle involving the edges $2$, $3$, and $6$ in $G(x)$, and thus there are paths from $s$ to $t$ in the subgraphs $G_2$, $G_3$, and $G_6$ in $\Gcyc(x)$.} \label{fig:Gcyc(x)}
\end{figure}

In order to use \cref{thm:stconn} to determine the query complexity of deciding $st$-connectivity on $\Gcyc(x)$, we next analyze the effective resistance (respectively capacitance) of $\Gcyc(x)$ in the presence (resp. absence) of cycles in $G(x)$. We first show the following relationship between effective resistance and circuit rank:

\begin{lemma}\label{lem:resCircRank}
Let $r$ be the circuit rank of $G(x)$. Then
\begin{align}\label{equation:reductRes}
R_{s,t}(\Gcyc(x)) =\frac{1}{r}.
\end{align}
\end{lemma}

The proof of \cref{lem:resCircRank} uses the following result from Ref. \cite{biggs1997} relating effective resistance and spanning trees:
\begin{theorem}\label{thm:resSpanTree}\cite{biggs1997}
Let $\{u,v\}=\ell$ be an edge in a connected graph $G$. Then the effective resistance between vertices $u$ and $v$
is equal to the number of spanning trees that include an edge $\ell$, divided by the total number of spanning trees:
\begin{align}
R_{u,v}(G) &=\frac{t_\ell(G)}{t(G)}.
\end{align}
\end{theorem}

\begin{proof}[Proof of \cref{lem:resCircRank}]
Using the rules that the effective resistance of graphs in series adds, (\cref{prop:ERseries}), and the inverse effective resistance of graphs in parallel adds, (\cref{prop:ERparallel}), we have: 
\begin{align}\label{equation:reductResorig}
R_{s,t}(\Gcyc(x)) =\left(\sum_{(u,v) \in E(G(x))} 1 - R_{u,v}(G(x))\right)^{-1}.
\end{align}
(We include this relatively straightforward calculation in \cref{app:RGcyc}.)

We next relate the righthand side of \cref{equation:reductResorig} to the circuit rank.
 Let $G(x)$ be a graph with $\kappa$ connected components. Let $g_i(x)$ be a subgraph
consisting of the $i\tth$ connected component of $G$, with $n_i$ vertices. We
count the number of times edges are used in all spanning trees of $g_i(x)$ in two
ways. First, we multiply the number of spanning trees by the number of edges
in each spanning tree. Second, for each edge we add the number of spanning
trees that include that edge. Setting these two terms equal, we have,
\begin{align}
t(g_i(x))(n_i-1) = \sum_{\ell \in E(g_i(x))} t_\ell(g_i(x)).
\end{align}

Rearranging, and using \cref{thm:resSpanTree} we have
\begin{align}\label{equation:connRes}
n_i-1 &= \sum_{\ell \in E(g_i(x))} \frac{t_\ell(g_i(x))}{t(g_i(x))} = 
\sum_{\{u,v\} \in E(g_i(x))} R_{u,v}(g_i(x)),
\end{align}
where if the sum has no terms (i.e. $E(g_i(x))=\varnothing$), we define it to be zero.

Summing over all $\kappa$ components of $G(x)$, we have
\begin{align}\label{eq:sumER}
n - \kappa=\sum_{\{u,v\} \in E(G)} R_{u,v}(G(x)),
\end{align}
Finally, using the fact that 
\begin{align}
\sum_{\{u,v\} \in E(G(x))}1=m,
\end{align}
where $m$ is the number of edges in $G(x)$, and combining with \cref{eq:sumER}, and \cref{def:circRank}, we have
\begin{align}\label{eq:r}
\sum_{\{u,v\}\in E(G(x))} 1- R_{u,v}(G(x)) = r.
\end{align}
Finally \cref{eq:r} and \cref{equation:reductResorig} give the result.

\end{proof}

We next analyze the effective capacitance of $\Gcyc(x)$ in the case of no cycles in $G(x)$: 
\begin{lemma}\label{lemma:effCapCycle}
If $G$ is the complete graph on $n$ vertices and $G(x)$ has no cycles and at most $\mu$ edges, then,
$C_{s,t}(\Gcyc(x))=O(n\mu^2)$.
\end{lemma}

\begin{proof}
We first analyze the effective capacitance of the subgraph $G_\ell(x)$ in two cases, when $\ell\in E(G(x))$ and when $\ell\notin E(G(x)).$

When $\ell\in E(G(x))$, then using \cref{prop:ECinf,prop:ECseries}, we have
$C_{s,t}(G_\ell(x))=C_{s,t}(G^-_\ell(x))$. Then using \cref{prop:ECcut}, we have that $C_{s,t}(G^-_\ell(x))$ is less than the size of the cut between vertices $s$ and $t$ in $G^-_\ell(x)$. Since there are $\mu$ edges and $n$ vertices in $G(x)$,
this quantity is bounded by $O(n\mu)$. (The worst case is when there are $\Omega(\mu)$ vertices connected to $s$, e.g.)

When $\ell\notin E(G(x))$, then using \cref{prop:noEdge,prop:ECseries},
$C_{s,t}(G_\ell(x))=O(1)$.

Since there are $n-\mu$ graphs $G_\ell(x)$ with $\ell\notin E(G(x))$ and $\mu$ graphs $G_\ell(x)$ with $\ell\in E(G(x))$, using \cref{prop:ECparallel} for graphs connected in parallel, we have that $C_{s,t}(\Gcyc(x))=O(\mu^2n).$
\end{proof}

In order to prove \cref{thm:circuitRankEst}, we
need to analyze $\widetilde{R}_-(x,\Gcyc)$:
\begin{lemma}\label{lemm:negWit}
For a graph $G$ with $m$ edges, let $X=\{x:G(x)\textrm{ contains a cycle and }
|E(G(x))|\leq \mu\}$ and let $\widetilde{R}_-=\max_{x\in
X}\widetilde{R}_-(x,\Gcyc)$. Then $\widetilde{R}_-=O(m\mu).$
\end{lemma}

\begin{proof}
Looking at \cref{eq:NegAppWitdef}, for $\ell\notin E(G(x))$, we have that all $v\in V(G_\ell)$ (except $s$ and $t$) are not in $V_x$, so any choice of $\sop V$ on these vertices will give an upper bound on the minimizing map. We choose $\sop V(v)=0$ for these vertices to give us our bound, which contributes $1$ to the sum for each such subgraph. Thus edges in these subgraphs contribute $m-\mu$ to the total.

For $\ell\in E(G(x))$, for vertices in these subgraphs which are also part of $V_x$, they will get mapped by $\sop V_x$ to values between $0$ and $1$ inclusive (since $\sop V_x$ can be seen as the voltage induced at each point by a unit potential difference between $s$ and $t$). If we choose the remaining vertices to also get mapped to values between $0$ and $1$ by $\sop V$, we will again have an upper bound on the minimum. Then $\left(\sop V(u)-\sop V(v)\right)^2\leq 1$ across all edges in these subgraphs. Since there are $m$ edges in each subgraph, and $\mu$ such subgraphs, edges in these subgraphs contributes $m\mu$ to the total.

Combining the two terms, we have that $\widetilde{R}_-(x,\Gcyc)\leq m-\mu+m\mu=O(m\mu).$
\end{proof}

Now we can put these results together to prove \cref{thm:circuitRankEst}:

\begin{proof}[Proof of \cref{thm:circuitRankEst}]
Using \cref{thm:witSizeEst}, \cref{lem:resCircRank}, and \cref{lemm:negWit}, we can estimate one over the circuit rank (i.e. $1/r$) to multiplicative error $\epsilon$. That is, we get an estimate of $1/r$ within $(1\pm\epsilon)/r$. Now if we take the inverse of this estimate, we get an estimate of of $r$ within $r/(1\pm\epsilon)$. But since $\epsilon\ll1$, taking the Taylor expansion, we have $1/(1\pm\epsilon)\approx(1\pm\epsilon)$ to first order in $\epsilon$.
\end{proof}

\section{Algorithms for Detecting Odd Paths, Bipartiteness, and Even Cycles}\label{section:Odd/Even Paths}

In this section, we note that a slight variation on one of the
st-connectivity problems considered by Cade et al. in Ref. \cite{cade2016time}
can be used to detect odd paths and bipartiteness; furthermore the
bipartiteness testing algorithm we describe is optimal in query complexity,
and far simpler than the bipartiteness algorithm in \cite{cade2016time}.  We
then use a similar construction to create a reduction from even-length cycle
detection to $st$-connectivity.

All of these algorithms involve the bipartite double graph of the original
graph. Given a graph $G$ with vertices $u$ and $v$, let $\Gd{u,v}$ be the
bipartite double graph of $G$, with vertex $u_0$ relabeled as $s$, and vertex
$v_1$ relabeled as $t$. To define $\Gd{u,v}(x)$, if $\{x,y\}\in E(G)$ is
associated with a literal, then $\{x_0,y_1\}$ and $\{x_1,y_0\}$ in
$E(\Gd{u,v})$ are associated with the same literal.

We first show a reduction from detecting an odd-length path to $st$-connectivty on $K^G$:
\begin{lemma}\label{lemm:oddReduct}
Let $G$ be a graph with vertices $u$ and $v$. There is an odd-length path from
$u$ to $v$ in $G(x)$ if and only if there is an $st$-path in $\Gd{u,v}(x)$
\end{lemma}

\begin{proof} 
Suppose there is an odd-length path from $u$ to $v$ in $G(x)$. Let the path be
$u,\eta^1,\eta^2,\dots,\eta^k,v$ where $k$ is an even integer greater than or equal to 0.
Then there is a path $s,\eta^1_1,\eta^2_0,\dots,\eta^k_0,t$, in
$\Gd{u,v}(x)$ (where the path goes through $\eta^i_{i\mod 2}$). For the other direction,
if there is a path from $s$ to $t$ in $\Gd{u,v}(x)$, there is an odd-length path from $u$
to $v$ in $G$. Note that any path in $\Gd{u,v}(x)$ must alternate between $0$- and
$1$-labeled vertices. If there is a path that starts at a $0$-labeled vertex
and ends at a $1$-labeled vertex, it must be an odd-length path. Then there
must be the equivalent path in $G(x)$, but without the labeling. See \cref{fig:OddPath} for an example of this reduction.
\end{proof}

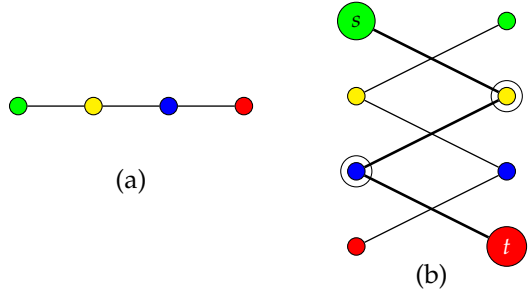
\begin{figure}[ht]
	\centering
    \begin{tikzpicture}[scale = 2]
	\node at (0,0)
		{\begin{tikzpicture}[scale = 2]
            \tikzstyle{edge}=[draw,line width = .5pt,-,black!100]
           	\tikzstyle{vert}=[draw,circle,scale=1]
           	\tikzstyle{vu}=[draw,circle,scale=.7,fill=green]
           	\tikzstyle{vv}=[draw,circle,scale=.7,fill=red]
           	\tikzstyle{vx}=[draw,circle,scale=.7,fill=blue]
           	\tikzstyle{vw}=[draw,circle,scale=.7,fill=yellow]
            \node[vu] (u) at (0,0) {};
            \node[vw] (w) at (.5,0) {};
            \node[vx] (x) at (1,0) {};
            \node[vv] (v) at (1.5,0) {};
            \draw[edge] (u) -- (w);
            \draw[edge] (w) -- (x);
            \draw[edge] (x) -- (v);
            \node[label] (a) at (.75,-.5) {(a)};
        \end{tikzpicture}};
    \node at (2,0) {
        \begin{tikzpicture}[scale = 2]
        	\tikzstyle{vert}=[draw,circle,scale=.9]
        	\tikzstyle{s}=[draw,rectangle,scale=1,fill=green]
        	\tikzstyle{t}=[draw,rectangle,scale=1,fill=red]
           	\tikzstyle{vu}=[draw,circle,scale=.7,fill=green]
           	\tikzstyle{vv}=[draw,circle,scale=.7,fill=red]
           	\tikzstyle{vx}=[draw,circle,scale=.7,fill=blue]
           	\tikzstyle{vw}=[draw,circle,scale=.7,fill=yellow]
            \tikzstyle{edge}=[draw,line width = .5pt,-,black!100]
            \tikzstyle{path}=[draw,line width = 1pt,-,black!100]
            \tikzstyle{halo}=[draw,circle,scale=1.25,fill=white]
            \node[halo] (h1) at (1,1) {};
            \node[halo] (h2) at (0,.5) {};
            \node[vert, fill=green, text=black] (u0) at (0, 1.5) {$s$};
            \node[vw] (w0) at (0, 1) {};
            \node[vx] (x0) at (0, .5) {};
            \node[vv] (v0) at (0, 0) {};
            \node[vu] (u1) at (1, 1.5) {};
            \node[vw] (w1) at (1, 1) {};
            \node[vx] (x1) at (1, .5) {};
            \node[vert, fill=red, text=white] (v1) at (1, 0) {$t$};
            \draw[path] (u0) -- (w1);
            \draw[edge] (w0) -- (x1);
            \draw[path] (x0) -- (v1);
            \draw[edge] (u1) -- (w0);
            \draw[path] (w1) -- (x0);
            \draw[edge] (x1) -- (v0);
            \node[label] (b) at (.5,-.2) {(b)};
        \end{tikzpicture}};
        \end{tikzpicture}

    \caption{(a) A simple example of a graph $G$ with an odd-length path between green and red vertices. (b) The bipartite double graph $K^G_{green, red}$ with a path between $s$ and $t$.} \label{fig:OddPath}
\end{figure}

\begin{theorem}\label{thm:bipartite}
Let $G$ be a graph with $n$ vertices and $m$ edges, with vertices $u$ and $v$.
Then there is a bounded-error quantum query algorithm that detects an odd
length path from $u$ to $v$ in $G$ using $O(\sqrt{nm})$ queries.
\end{theorem}

\begin{proof}
Using \cref{lemm:oddReduct} we reduce the problem to $st$-connectivity on
$\Gd{u,v}$. Then using \cref{thm:stconn}, we need to bound the largest
effective resistance and effective capacitance of $\Gd{u,v}(x)$ for any string $x$. The longest
possible path from $s$ to $t$ in $\Gd{u,v}(x)$ is $O(n)$ so by \cref{prop:length},
$R_{s,t}(\Gd{u,v}(x))=O(n).$ The longest possible cut between $s$ and $t$
is $O(m)$, so by \cref{prop:ECcut}, $C_{s,t}(\Gd{s,t}(x))=O(m).$ This gives the
claimed query complexity.
\end{proof}


Note $u_0$ is connected to $u_1$ in $K^G(x)$ if and only if there is
an odd-length path from $u$ to itself in $G(x),$ where this path is allowed to
double back on itself, as in \cref{fig:OddCycle}. This odd-length path in turn occurs if and only if the
connected component of $G(x)$ that includes $u$ is not bipartite (has an odd
cycle)! Thus if we are promised that $G(x)$ is connected, we can pick any
vertex in $G$, run the algorithm of \cref{thm:bipartite} on $\Gd{u,u}(x)$, and
determine if the graph is bipartite, which requires $O(\sqrt{nm})$ queries.

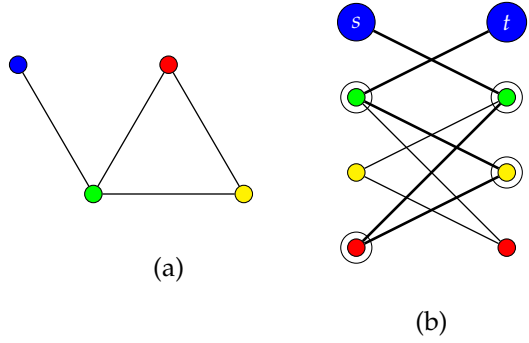
\begin{figure}[ht]
    \centering
    \begin{tikzpicture}[scale = 2]
	\node at (0,0)
		{\begin{tikzpicture}[scale = 2]
            \tikzstyle{edge}=[draw,line width = .5pt,-,black!100]
           	\tikzstyle{vert}=[draw,circle,scale=1]
           	\tikzstyle{vu}=[draw,circle,scale=.7,fill=green]
           	\tikzstyle{vv}=[draw,circle,scale=.7,fill=red]
           	\tikzstyle{vx}=[draw,circle,scale=.7,fill=blue]
           	\tikzstyle{vw}=[draw,circle,scale=.7,fill=yellow]
            \node[vu] (u) at (0,0) {};
            \node[vw] (w) at (1,0) {};
            \node[vv] (v) at (.5,.86) {};
            \node[vx] (x) at (-.5, .86) {};
            \draw[edge] (u) -- (v);
            \draw[edge] (v) -- (w);
            \draw[edge] (x) -- (u);
            \draw[edge] (w) -- (u);
            \node[label] (a) at (.5,-.5) {(a)};
        \end{tikzpicture}};
    \node at (2,0) {
        \begin{tikzpicture}[scale = 2]
        	\tikzstyle{vert}=[draw,circle,scale=.9]
           	\tikzstyle{vu}=[draw,circle,scale=.7,fill=green]
           	\tikzstyle{vv}=[draw,circle,scale=.7,fill=red]
           	\tikzstyle{vx}=[draw,circle,scale=.7,fill=blue]
           	\tikzstyle{vw}=[draw,circle,scale=.7,fill=yellow]
            \tikzstyle{edge}=[draw,line width = .5pt,-,black!100]
            \tikzstyle{path}=[draw,line width = 1pt,-,black!100]
            \tikzstyle{halo}=[draw,circle,scale=1.25,fill=white]
            \node[halo] (h1) at (0,0) {};
            \node[halo] (h2) at (1,.5) {};
            \node[halo] (h3) at (0,1) {};
            \node[halo] (h4) at (1,1) {};
            \node[vert, fill=blue,text=white] (x0) at (0, 1.5) {$s$};
            \node[vu] (u0) at (0, 1) {}; 
            \node[vw] (w0) at (0, .5) {}; 
            \node[vv] (v0) at (0, 0) {};
            \node[vert,fill=blue,text=white] (x1) at (1, 1.5) {$t$};
            \node[vu] (u1) at (1, 1) {};
            \node[vw] (w1) at (1, .5) {};
            \node[vv] (v1) at (1, 0) {};
           	\draw[path] (x0) -- (u1);
           	\draw[path] (x1) -- (u0);
            \draw[edge] (u0) -- (v1);
            \draw[edge] (v1) -- (w0);
            \draw[edge] (w0) -- (u1);
            \draw[path] (u1) -- (v0); 
            \draw[path] (v0) -- (w1); 
            \draw[path] (w1) -- (u0);
            \node[label] (b) at (.5,-.5) {(b)};
        \end{tikzpicture}};
        \end{tikzpicture}
    \caption{(a) A graph $G$ with an odd cycle. (b) The bipartite double graph $K^G$ with a path between the two green vertices.} \label{fig:OddCycle}
\end{figure}

On the other hand, if we are not promised that $G(x)$ is connected, we simply
need to check whether there is an odd path from any of the $n$ vertices of $G$
to itself, and we now show that doing this check does not increase the query
complexity. We use a similar strategy as with cycle detection:

\begin{theorem}
Let $G$ be a graph with $n$ vertices and $m$ edges. Then there is a bounded-error quantum query algorithm that detects an odd cycle (in effect, non-bipartiteness) in $O(\sqrt{nm})$ queries.
\end{theorem}

\begin{proof}[Proof Sketch]
Let $\Gbi$ be the graph that consists of the the graphs $\Gd{u,u}$ composed in
parallel for all $u\in V(G)$. This amounts to evaluating the logical
\textsc{or} of there being an odd cycle connected to any vertex in $G$.

A similar analysis as in cycle detection shows that the effective resistance
of $\Gbi$ is $O(1)$, if there is an odd cycle. On the other hand, since there
are $n$ copies of $K^G$ in this new graph, and each copy has $m$ edges, the
largest possible cut is $O(nm)$. Applying \cref{thm:stconn} gives
the result.
\end{proof}

\begin{figure}[h!]
    \centering
    \begin{tikzpicture}[scale = 2]
	\node at (0,0)
		{\begin{tikzpicture}[scale = 2]
            \tikzstyle{edge}=[draw,line width = .5pt,-,black!100]
           	\tikzstyle{vert}=[draw,circle,scale=1]
           	\tikzstyle{vu}=[draw,circle,scale=.7,fill=green]
           	\tikzstyle{vv}=[draw,circle,scale=.7,fill=red]
           	\tikzstyle{vx}=[draw,circle,scale=.7,fill=blue]
           	\tikzstyle{vw}=[draw,circle,scale=.7,fill=yellow]
            \node[vu] (u) at (0,1) {};
            \node[vw] (w) at (0,0) {};
            \node[vx] (x) at (1,0) {};
            \node[vv] (v) at (1,1) {};
            \draw[edge] (u) -- (v);
            \draw[edge] (v) -- (x);
            \draw[edge] (x) -- (w);
            \draw[edge] (w) -- (u);
            \node[label] (l) at (.5,1.2) {$\ell$};
            \node[label] (a) at (.5,-.5) {(a)};
        \end{tikzpicture}};
    \node at (2,0) {
        \begin{tikzpicture}[scale = 2]
        	\tikzstyle{vert}=[draw,circle,scale=.9]
           	\tikzstyle{vu}=[draw,circle,scale=.7,fill=green]
           	\tikzstyle{vv}=[draw,circle,scale=.7,fill=red]
           	\tikzstyle{vx}=[draw,circle,scale=.7,fill=blue]
           	\tikzstyle{vw}=[draw,circle,scale=.7,fill=yellow]
            \tikzstyle{edge}=[draw,line width = .5pt,-,black!100]
            \tikzstyle{path}=[draw,line width = 1pt,-,black!100]
            \tikzstyle{halo}=[draw,circle,scale=1.25,fill=white]
            \node[halo] (h2) at (0,1.5) {};
            \node[halo] (h3) at (1,1) {};
            \node[halo] (h4) at (0,.5) {};
            \node[vert] (s) at (0, 2) {$s$};
            \node[vu] (u0) at (0, 1.5) {};
            \node[vw] (w0) at (0, 1) {};
            \node[vx] (x0) at (0, .5) {};
            \node[vv] (v0) at (0, 0) {};
            \node[vu] (u1) at (1, 1.5) {};
            \node[vw] (w1) at (1, 1) {};
            \node[vx] (x1) at (1, .5) {};
            \node[vert, fill=red, text=white] (v1) at (1, 0) {$t$};
            \draw[path] (s) -- (u0);
            \draw[path] (u0) -- (w1);
            \draw[edge] (u1) -- (w0);
            \draw[path] (w1) -- (x0);
            \draw[edge] (w0) -- (x1);
            \draw[path] (x0) -- (v1);
            \draw[edge] (x1) -- (v0);
            \node[label] (l) at (-.2,1.75) {$\ell$};
            \node[label] (b) at (.5,-.5) {(b)};
        \end{tikzpicture}};
        \end{tikzpicture}
    \caption{(a) An simple example of a graph $G$ with an even-length cycle. (b) The bipartite double graph $K^G_{green,red}$ connected in series with $G^1_{\{red,green\}}$. We see there is a path from $s$ to $t$ in this graph, corresponding to an even-length cycle passing through $\ell$.} \label{fig:Kg(x)}
\end{figure}
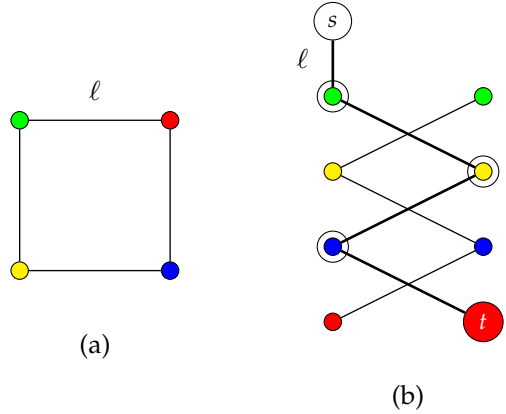

Finally, we show how to detect even cycles:
\begin{theorem}
Let $G$ be a graph with $n$ vertices and $m$ edges. Then there is a bounded-error quantum query algorithm that detects an even-length cycle in $O(\sqrt{nm})$
queries.
\end{theorem}

\begin{proof}
For an edge $\ell=\{u,v\}\in E(G)$, note that there is an $st$-path in
$K^{G^-_\ell}_{u,v}$ if and only if there is an odd-length path from $u$ to
$v$ that does not use the edge $\ell$ itself. Thus if we consider the graph
composed of $G^1_\ell$ and $K^{G^-_\ell}_{u,v}$ in series, which we denote $G^E_\ell$,
there is an $st$-path if and only if there is an even-length cycle through
$\ell$. Finally, if we compose the graphs $G^E_\ell$ in parallel for all
$\ell\in E(G)$, we obtain a graph that has an $st$-path if and only if there
is an even cycle passing through some edge in $G$, as in \cref{fig:Kg(x)}.

As in our previous analyses of cycle detection and bipartiteness, if there is
an even cycle, the effective resistance will be $O(1).$ On the other hand, if
there is no even-length cycle, then it is a fairly well known fact that the number of
edges in $G$ is $O(n)$. Then similar to previous analyses, for each graph
$G^E_\ell$ such that $\ell\in E(G)$, we have that the cut is $O(m)$.
Otherwise, for each graph $G^E_\ell$ such that $\ell\in E(G)$, we have that
the cut is $O(1)$. Thus a bound on the size of the total cut is
$O(n^2+nm)=O(nm)$ (assuming that $n=O(m)$.) Applying \cref{thm:stconn} gives
the result.
\end{proof}


\section{Acknowledgments}

This research was sponsored by the Army Research Office
and was 
accomplished under Grant Number 
W911NF-18-1-0286.
The views and conclusions contained in this 
document are those of the authors and should not be interpreted as representing the official policies, either 
expressed or implied, of the Army Research Office
or the U.S. Government.  The U.S. Government is 
authorized to reproduce and distribute reprints for Government purposes notwithstanding any copyright 
notation herein.



\bibliography{connectivity}
\bibliographystyle{plain}

\appendix
\section{Effective Resistance of $\Gcyc$}\label{app:RGcyc}

We relate the effective resistance across edges in
$G$ to the effective resistance
between $s$ and $t$ in $\Gcyc$. (The proof also applies to $G(x)$ and $\Gcyc(x)$.) We will write $R_{s,t}(\Gcyc)$
in terms of a sum of $R_{u,v}(G)$ where $(u,v,\ell)$ is an edge on a cycle 
in $G$.

Consider an edge $\{u,v\}\in E(G)$. Then using \cref{prop:ERparallel} (for graphs composed in parallel), we have
\begin{align}
\frac{1}{R_{u,v}(G)} 
= 1 + \frac{1}{R_{s,t}(G^-_\ell)}.
\end{align}
Rearranging, we have
\begin{align}
R_{s,t}(G^-_\ell) = \frac{R_{u,v}(G)}{1 -R_{u,v}(G)}
\end{align}
Then using \cref{prop:ERseries} (for graphs composed in series), we have that 
\begin{align}
R_{s,t}(G_\ell)&=R_{s,t}(G^-_\ell) + 1 \nonumber\\
&= 
\frac{R_{u,v}(G)}{1 - R_{u,v}(G)} + 1 \nonumber\\
&= \frac{1}
{1-  R_{u,v}(G)}.
\end{align}
Finally, using \cref{prop:ERparallel} (graphs composed in parallel) again, we have
\begin{align}
\frac{1}{R_{s,t}(\Gcyc)} = \sum_{\{u,v,\} \in E(G)} 
1- R_{u,v}(G).
\end{align}

\end{document}